\documentclass[12pt, a4paper]{article}
\usepackage[utf8]{inputenc}
\usepackage[T1]{fontenc}
\usepackage[margin=1.25in]{geometry}
\usepackage{xcolor}
\usepackage{natbib}
\setcitestyle{authoryear,open={(},close={)}}
\usepackage{titlesec}
\titleformat*{\section}{\large\bfseries}
\titleformat*{\subsection}{\normalsize\bfseries}
\titleformat*{\subsubsection}{\normalsize\bfseries}

\usepackage{caption}   

\usepackage{comment}
\usepackage{xcolor}
\usepackage{tikz}
\usetikzlibrary{patterns}
\usepackage{graphics}
\usepackage{graphicx}
\usepackage{color}
\usepackage{amsmath}
\usepackage{amsthm}
\usepackage{amsfonts}
\usepackage{geometry}
\usepackage[normalem]{ulem}
\usepackage{float}
\usepackage{dsfont} 

\usepackage{mathabx} 

\usepackage{hyperref}
\usepackage{cleveref}
\hypersetup{colorlinks=true,linkcolor=blue,citecolor=blue}
\usepackage{dsfont} 
\usepackage{nicefrac}

\date{}

\definecolor{ForestGreen}{rgb}{.13,.54,.13}
\definecolor{violet}{cmyk}{0.79,0.88,0,0}
\definecolor{darkBrickRed}{rgb}{.70,.13,.16}

\ifdefined\DRAFT

\newcommand{\fed}[1]{{\color{ForestGreen}{(\textbf{Fedor:} #1)}}}

\else
\newcommand{\fed}[1]{}

\fi

\newtheorem{theorem}{Theorem}
\newtheorem*{theorem*}{Theorem}
\newtheorem{lemma}{Lemma}
\newtheorem{claim}{Claim}

\newtheorem{proposition}{Proposition}
\newtheorem*{proposition*}{Proposition}

\theoremstyle{definition}

\theoremstyle{remark}

\def\eps{\varepsilon}
\def\cus{\mathrm{cousin}}
\def\kill{\mathrm{killer}}

\usepackage{xparse}
\usepackage{mleftright}

\DeclareDocumentCommand\Pr{ m g }{\ensuremath{
    {   \IfNoValueTF {#2}
      {\mathbb{P}\mleft[{#1}\mright]}
      {\mathbb{P}\mleft[{#1}\,\middle\vert\,{#2}\mright]}%
    }
}}
\DeclareDocumentCommand\E{ m g }{\ensuremath{
    {   \IfNoValueTF {#2}
      {\mathbb{E}\mleft[{#1}\mright]}
      {\mathbb{E}\mleft[{#1}\,\middle\vert\,{#2}\mright]}%
    }
}}

\title{Decentralized Equilibrium for Bitcoin Mining\footnote{We thank Matt Weinberg for useful comments and suggestions.}}

\author{ \large Manuel Mueller-Frank\thanks{IESE Business School. Email: mmuellerfrank@iese.edu. Manuel Mueller-Frank was supported by an IESE High Impact grant.} \ \ \ %
Minghao Pan\thanks{Caltech. Email: mpan2@caltech.edu.} \ \ \ Omer Tamuz\thanks{Caltech. Email: tamuz@caltech.edu.  Omer Tamuz was supported by a National Science Foundation CAREER award (DMS-1944153) and a MURI grant.}}

\begin{document}

\maketitle

\begin{abstract}
    Cryptocurrencies such as Bitcoin are defined by protocols that specify how participants record transactions and create new currency units. These protocols are not enforced by law or any central entity and instead are intended to be incentive compatible. However, the Bitcoin mining protocol proposed by \cite{nakamoto2008Bitcoin} and implemented in practice is known not to constitute an equilibrium  \citep{eyal2018majority}. This leaves open the question of whether the decentralized outcome intended by Nakamoto can be sustained in equilibrium in the Bitcoin mining game. We propose inertial mining, a novel mining protocol that induces that outcome, i.e., a single longest chain in which each miner's asymptotic share of blocks equals its share of computational power. Our main result establishes that inertial mining constitutes an equilibrium, assuming no miner controls one half or more of the computational power. Inertial mining coincides with Nakamoto's protocol on the equilibrium path, and can be implemented in Bitcoin without any changes to its consensus mechanism or blockchain architecture. When a single miner controls more than half of the computational power, we show that no decentralized equilibrium exists.

    
   \end{abstract}

\section{Introduction}
\cite{nakamoto2008Bitcoin} introduced Bitcoin, a digital money that functions without a central entity controlling the monetary base and intermediating transactions.  Although Bitcoin has so far failed to gain significant traction as a medium of exchange, it has become a major alternative financial asset and a digital store of value. Its market capitalization first surpassed \$1 billion in 2013, \$100 billion in 2017, and \$1 trillion in 2021.\footnote{See www.coingecko.com/en/coins/Bitcoin.} There is a  growing literature that analyzes Bitcoin from an economic and game-theoretic perspective. See, for example, \cite*{biais2019blockchainfolk}, \cite{SchillingUhlig2019SimpleBitcoinEconomics}, \cite{leshno2020Bitcoin}, \cite{huberman2021monopoly},  \cite{pagnotta2022decentralizing} and \cite{bahrani2024undetectable}. For a survey, see \cite{HalaburdaHaeringerGansGandal2022MicroeconomicsCrypto}.

The main feature that makes Bitcoin attractive as a reserve asset is its decentralized nature, which implies that no entity, governmental or otherwise, can censor transactions. Decentralization, as well as the absence of a central intermediary that selects and processes valid transactions, introduces two major challenges. First, it requires consensus among all participants about the allocation of Bitcoin ownership.
Second, it requires that all entities participating in the transaction selection and execution follow the protocol, rather than trying to increase their profits by deviating from it. Satoshi Nakamoto addressed these complications via the so-called \emph{Nakamoto consensus} mechanism. 

The basic idea of this mechanism is that miners can create \emph{blocks} at a computational cost. These blocks, which contain transactions, are chained together, with the resulting blockchain serving two functions: it assigns Bitcoin ownership to miners and sequentially records transactions. Nakamoto’s \emph{Bitcoin mining protocol} describes how miners determine which blocks form part of the established transaction history. It is simple, specifying that miners always add blocks to the end of the longest chain and that, among competing chains, the longest chain is treated as the established transaction history, whereas blocks on the other chains are not included in that history. Competing chains can either be introduced by strategic miners, or inadvertently through \emph{network outages}: episodes in which internet traffic between different parts of the network is slowed or halted. 

The initial understanding was that  the Bitcoin mining protocol constitutes an equilibrium for the participants.  \cite{eyal2018majority}, however, establish that this standard protocol (or strategy profile) is not an equilibrium, and construct a profitable deviation strategy they call ``selfish mining''. 
This gives rise to a paradox: Bitcoin's core distinguishing feature is its decentralization, yet its prescribed protocol is subject to profitable deviations. To the best of our knowledge, no equilibrium has been shown to implement the decentralized outcome that Bitcoin is intended to produce (but some related work is discussed below). 
Thus, even at the level of equilibrium existence, it remains unclear whether Bitcoin can achieve its defining promise of decentralization. Does there exist a \emph{fair, decentralized equilibrium} in which multiple miners contribute to the longest chain and each miner's share of blocks accurately reflects its computational mining power?

The contribution of this paper is to resolve this paradox. We propose the \emph{inertial mining protocol}, which constitutes an equilibrium for the Nakamoto consensus mechanism, coincides with the Nakamoto protocol on-path, and thus produces the same outcome.

Following the computer science literature, we model the interaction between miners as the following mining game. Each miner $i$ controls a fraction $\alpha_i$ of the total computational mining power. Time is discrete, and at time $0$ an initial \emph{genesis block} comes into existence. In each subsequent time period, each miner tries to mine a new block, and one of the miners is selected by nature to succeed, with probability proportional to their mining power. A mined block is attached to one of the previously mined blocks, with miners choosing where to attach their blocks. A miner can choose to extend the currently longest chain (thus following the standard Bitcoin protocol) or attach blocks elsewhere, potentially creating competing chains. Miners can also choose when to publish their blocks, making them visible to other miners.

We assume that miners are long-lived and adopt the standard assumption that the payoff of a given miner is equal to the share of blocks he mined in the longest chain. Note that if all miners follow the standard mining protocol, then all blocks are consecutively mined on one chain and asymptotically each miner $i$'s payoff is equal to his mining power $\alpha_i$. This is the fair outcome which the Nakamoto consensus mechanism strives to achieve.

If a miner $i$ has mining power $\alpha_i > \frac12$, the standard Bitcoin mining protocol is not an equilibrium, for a trivial reason: that miner can eventually overtake any public chain by extending his own private branch.  \cite{eyal2018majority} showed that there is a profitable deviation even if $\alpha_i$ is below one half (but not too low, see more below). The core idea is that a \emph{selfish miner} withholds mined blocks, privately builds a parallel chain to the public longest chain, and selectively discloses it if his private chain generates a lead over the public one. This way the miners that follow the standard mining protocol waste resources and the selfish miner ends up with a payoff strictly larger than his computational power.

Our contribution is to construct an   equilibrium such that no miner has an incentive to deviate, assuming $\max_i \alpha_i < 1/2$. The inertial mining protocol we propose works as follows. Consider miners who are working on the last block in the current longest chain. If a new block gets appended to this chain, the miners will switch to working on the new last block, as in the standard mining protocol. However, if a new chain is published that does not extend the current longest chain (as happens in the selfish mining deviation), the miners will only switch to it if it is longer by at least $I>0$ blocks. The number $I$ is a parameter of the inertial mining protocol. To ensure that this is an equilibrium, $I$ needs to be chosen to be large enough, as a function of  $\max_i\alpha_i$. The closer this is to one half, the larger $I$ needs to be. When there is a miner with power greater than one half, no choice of $I$ makes inertial mining an equilibrium. Indeed, when $\max_i\alpha_i$ is greater than one half, we show that every equilibrium is centralized (Proposition~\ref{proposition:onlymonopoly}).

Inertial mining differs from the standard mining protocol in its prescription of which chain to append to if there is more than one public chain, but results in a single chain as the equilibrium outcome. Thus, it achieves the intended outcome of the standard mining protocol. Importantly, it does this robustly as an equilibrium, with miners having no incentive to deviate. It is straightforward to see that under inertial mining, selfish mining is no longer a profitable deviation. The main technical contribution of this paper is to show that no other possible deviation is profitable. This is challenging because the set of possible deviations is large. 

Our proof addresses this by introducing an accounting system that assigns to each block that is strategically mined by the deviating player a set of blocks of the other players that were ``killed'' by that block. A strategically mined block can benefit the deviating player by killing blocks of the other players. However, such a block necessarily involves a risk: it is not mined honestly, and so may not end up as part of the longest chain, and be lost to the deviating player. Our proof shows that no killer block can be expected to displace enough honest blocks to offset the cost of the event that this block does not enter the longest chain. 

\subsection{Related literature}

This paper contributes to the growing literature on the incentive properties and equilibrium analysis of proof-of-work (PoW) blockchains such as Bitcoin. Our starting point is the vulnerability first identified by \cite{eyal2018majority}: Bitcoin’s standard mining protocol does not constitute an equilibrium, since a miner with less than majority mining power may profitably deviate by withholding blocks and releasing them strategically. This selfish mining strategy creates intentional and persistent forks.

\cite*{optimalselfishmining} strengthen this critique by characterizing optimal withholding and disclosure strategies against honest miners following the standard protocol and showing that such strategies can strictly dominate the original selfish-mining deviation. The fact that the canonical Bitcoin mining protocol is strategically fragile motivates the central question we study: whether there exists an equilibrium of a PoW blockchain that generates a single longest chain on the equilibrium path and where each miner's share of blocks equals his mining power, as intended by \cite{nakamoto2008Bitcoin}. Our main result answers this question affirmatively.

Existing equilibrium analyses of mining games consider more restrictive information structures or strategy spaces. 
\cite{kiayias2016blockchain} analyze a complete information Bitcoin mining game in which each block discovery and the identity of the successful miner immediately become common knowledge. They restrict each miner to build at most one block at each depth, which reduces the relevant state to a two-branch race between the honest branch and a deviating branch. They show that the standard Bitcoin protocol is an equilibrium when all $\alpha_i$ are smaller or equal to $0.308$. For a game that disallows miners to strategically withhold and disclose blocks, \cite{biais2019blockchainfolk} have established a decentralized Markov Perfect equilibrium. Relatedly, \cite{pagnotta2022decentralizing} embeds PoW mining incentives in a general-equilibrium environment in which the cryptocurrency’s price and security are jointly determined. He shows that multiple equilibria can arise with different price-security combinations. 

A closely related strand of work studies strategic mining in PoW blockchains, but proposes to eliminate selfish mining by changing the consensus mechanism or the blockchain design itself. \cite{heilman2014oneweird} proposes the \emph{Freshness Preferred} mining protocol, under which miners break ties by favoring blocks with more recent verifiable timestamps, thereby reducing the profitability of delayed block release and selfish mining. Because this approach relies on introducing unforgeable timestamps as a novel feature, it necessitates a redesign of blockchain. \cite{solat2017zeroblock} propose changing Bitcoin’s block-validation and chain-extension rules by requiring honest miners to create and mine on a special “ZeroBlock” whenever no valid block arrives within a predetermined interval based on expected block-finding and propagation times. \cite{pass2017fruitchains} propose a redesigned blockchain architecture, FruitChains, that weakens the ability of strategic miners to benefit from fork manipulation and thereby improves robustness to selfish mining. While these contributions demonstrate that selfish-mining incentives can be mitigated, they all rely on changing the consensus mechanism and/or the architecture of the blockchain itself. By contrast, our approach does not modify Bitcoin's underlying architecture or mechanism; instead, we propose an alternative mining protocol which can be implemented in the same technical environment as the standard Bitcoin protocol.


More broadly, our analysis relates to the literature on attacks against Nakamoto consensus and on the economic limits of PoW security. A related strand studies double-spend attacks; see, for example, \cite{Bonneau2016WhyBuyRent} and \cite{GansHalaburda2024ZeroCostMajority}. \cite{budish2025trust} argues that the permissionless security of Nakamoto consensus is intrinsically costly because the flow cost of security must scale with the value at risk from attack. \cite*{leshno2024viability} develop an alternative consensus design that preserves permissionless entry while delivering stronger economically meaningful security guarantees. These papers focus on different vulnerabilities or on alternative institutional designs, whereas our focus is the equilibrium resolution of selfish mining within PoW.

Finally, the motivation for our analysis is further strengthened by recent work emphasizing that selfish mining is not merely a theoretical possibility.  In its more than sixteen years of operation, little selfish mining activity has been observed on Bitcoin.  However, \cite*{li2024statistical} provide empirical evidence consistent with selfish-mining behavior in several other PoW blockchains, with especially strong evidence for Monacoin and Bitcoin Cash. More conceptually, the absence of a known equilibrium in PoW that implements a single longest chain, which satisfies decentralized consensus on the equilibrium path, has been emphasized in discussions surrounding Ethereum’s transition away from PoW; see \cite{buterin2017posfaq} and \cite*{hall2024study}. Our contribution to this debate is to show that such an equilibrium does exist: inertial mining sustains a single decentralized longest chain on the equilibrium path without requiring any modification of Bitcoin’s consensus mechanism or blockchain architecture.

\par
\par

\section{The Model}

We begin by formally defining the concept of a blockchain and consensus, and then introduce the underlying game that captures the strategic interactions involved in Bitcoin mining. This game is widely used in the computer science literature \citep{eyal2018majority,optimalselfishmining, bahrani2024undetectable}.

\subsection{Blocks and chains}
A \emph{block} is a pair $(x,y)$, where $x,y$ take values in a set of labels, which we take to be the unit interval $[0,1]$. The label of the block $(x,y)$ is $x$, and $y$ is the label of the block's parent block. The block $(0,0)$ is called the \emph{genesis} block. A (block)chain is a finite  sequence of blocks $C = (x_1,y_1),(x_2,y_2),\ldots,(x_n,y_n)$ such that $x_i \neq x_j$ for $i \neq j$, $(x_1,y_1)=(0,0)$, and for $i>1$ it holds that $y_i = x_{i-1}$. We say that $(x_i,y_i)$ is the \emph{predecessor} block of $(x_{i+1},y_{i+1})$, and that the latter is a \emph{successor} of the former. For $i < j$ we say that $(x_i,y_i)$ is an \emph{ancestor} of $(x_j,y_j)$, and the latter is a \emph{descendant} of the former.

\subsection{The Bitcoin Mining Game}

There is a finite set of players $N$. Each player $i$ is exogenously assigned \emph{mining power} $\alpha_i > 0$, with $\sum_i\alpha_i = 1$. There are discrete time periods $t \in \{1,2,\ldots\}$.

Each player $i$ has a finite set of \emph{mined blocks} $M_t^i$ at the end of period $t$. We denote by $M_t = \cup_i M_t^i$ the set of all mined blocks. There is a set of \emph{public blocks} $P_t$ at the end of period $t$, which is a subset of $M_t$. At the beginning of period $t$, player $i$ can observe their own past blocks $M_{t-1}^i$ as well as the past public blocks $P_{t-1}$, but not the other players' mined blocks. Thus, the history available to player $i$ at the beginning of period $t$ consists of their own actions up to that point, and in addition $P_0,P_1,\ldots,P_{t-1},M_0^i,M_1^i,\ldots,M_{t-1}^i$. Beyond this, players do not observe the actions of other players.

We set $M_0^i = \emptyset$ for all $i \neq i_0$, and $M_0^{i_0} = \{(0,0)\}$. That is, at the start of the game players have no mined blocks, except for some player $i_0$ who has the genesis block $(0,0)$. We also set $P_0 = \{(0,0)\}$, so that this block is public.\footnote{The choice of $i_0$ is immaterial, and in fact this block could be unowned; we assign it ownership for notational convenience.}

At each time period nature chooses one of the players, where the probability that $i$ is chosen is $i$'s mining power $\alpha_i$. Nature's choices are independent across periods. 

At the beginning of each time period each player $i$ has to choose an action $b_t^i \in [0,1]$. As we shall see, this will be interpreted as the label of the block to which the player wants to add a block. If player $i$ is chosen at period $t$, they mine a new block $m_t = (x,b^i_t)$, where $x$ is chosen independently and uniformly at random from the set of labels $[0,1]$. That is, if they mine $m_t = (x,b^i_t)$ then $M_{t}^i = M_{t-1}^i \cup \{m_t\}$. Otherwise, $M_t^i = M_{t-1}^i$. Note that since block labels are chosen uniformly at random from $[0,1]$, each block will almost surely have a unique label.

Besides choosing $b^i_t$, players can, at the end of any time period, decide to \emph{publish} any blocks in $M_t^i$. That is, each player chooses a subset $B_t^i \subseteq M_t^i$, and we set $P_t =  P_{t-1} \cup \left(\cup_i B_t^i\right)$.

In addition, players observe a public randomization device: an i.i.d.\ process $(\Xi_t)_t$. This will be useful to coordinate on tie-breaking.\footnote{In practice, the chain itself can be used as such a device, by the use of standard secret-sharing techniques, applied to the block labels.} In summary, in period $t$, player $i$ observes $\Xi_1,\ldots,\Xi_{t}$, $P_1,\ldots,P_{t-1}$, and $M_1^i,M_2^i,\ldots,M_{t-1}^i$; chooses $b_t^i$; then, with probability $\alpha_i$, mines a block $(x,b_t^i)$, which is added to $M_{t-1}^i$ to form $M_t^i$; and finally can choose to publish any subset $B_t^i$ of $M_t^i$. This defines the strategy space of the players in the mining game. We next define their utilities.

Recall that $P_t$ is the set of published blocks at time $t$. Since $P_t$ includes the genesis block, it includes at least one chain (starting from the genesis block), and it furthermore includes a longest chain. Let $C_t^1,\ldots,C_t^k \subseteq P_t$ be the longest chains in $P_t$. If there is a unique longest chain $C_t^1$, we denote $C_t=C_t^1$, and set the utility of player $i$ to be the (asymptotic) fraction of blocks it owns in $C_t$:
\begin{align*}
    u_i = \liminf_t \frac{|C_t \cap M_t^i|}{|C_t|}.
\end{align*}
If there is more than one longest chain, we average over them:
\begin{align*}
    u_i=\liminf_t \frac{1}{k}\sum_{\ell=1}^k\frac{|C_t^\ell \cap M_t^i|}{|C_t^\ell|}.
\end{align*}
Utility is defined over relative rather than absolute block rewards. This normalization is standard in the literature and reflects the fact that the unit of account for block production is technologically arbitrary: a protocol change can rescale the number of blocks, just as a stock split can rescale the number of shares, without altering the miner’s underlying economic payoff. More generally, our choice of utility reflects an assumption that miners are long-lived; indeed, due to immense fixed and variable costs, Bitcoin mining is conducted primarily by companies (and some governments) rather than individuals. The computer science literature likewise focuses on the asymptotic fraction rather than (say) the discounted fraction,\footnote{See \cite*{eyal2018majority,optimalselfishmining, bahrani2024undetectable}.} but we do expect that similar results hold for more standard discounted utilities. As we are motivated by the fundamental open question implied by the results of \cite{eyal2018majority}, 
we naturally use their utility function.\footnote{The only modification relative to \cite{eyal2018majority} is our focus on the limit inferior of the fraction of own blocks in the longest chain as opposed to the limit. This is necessary as for general strategies the limit of the fraction might not exist.}

For a given mixed strategy profile $\sigma$, we consider player $i$'s expected utility $U_i(\sigma)=E_\sigma[u_i]$. We say that a strategy profile $\sigma$ is an \emph{equilibrium} if $U_i(\sigma)\ge U_i(\sigma_i',\sigma_{-i})$ for every player $i$ and strategy $\sigma'_i$.

To conclude this discussion, we define some relevant properties of an outcome. We say that an outcome satisfies \emph{decentralized consensus} if there are at least two players $i,j$ such that 
$u_i,u_j>0$. If $u_i=1$ for some player, the outcome satisfies \emph{centralized consensus}. The outcome satisfies \emph{fair decentralized consensus} if for each player $i$, his asymptotic share of blocks $u_i$ equals his mining power $\alpha_i$.

\subsection{Honest Mining, Selfish Mining and an Open Problem}
In the original Bitcoin whitepaper, \cite{nakamoto2008Bitcoin} prescribes a mining strategy under which miners extend the longest public chain and immediately publish newly discovered blocks. Following \cite{eyal2018majority}, we refer to this mining strategy as \emph{honest mining}. Formally, $b_t^i$ is chosen to be the label of the last block in the longest chain in $P_{t-1}$. If there are multiple longest chains, then players choose uniformly at random among the last blocks of the longest chains, using the public randomization device. Whenever a player mines a block, they immediately publish it. It is straightforward to see that all players adopting honest mining results in a single infinite chain that almost surely satisfies fair decentralized consensus.\footnote{This follows directly from a strong law of large numbers.} 

For fair decentralized consensus to be a robust outcome of the Bitcoin mining game for a mining power profile $\alpha$, the honest mining strategy profile needs to be an equilibrium, i.e., a mutual best-response. However, the seminal paper of \cite{eyal2018majority} proves that the honest mining strategy profile is not an equilibrium. They introduce \emph{selfish mining}, a profitable deviation strategy, which we describe now.

Suppose all players follow the protocol, except $i$, who implements selfish mining. Consider the first time $t$ in which $i$ finds a block. Then up to time $t$ all found blocks form one chain. The selfish miner $i$ appends his block to the last block of the longest public chain but does not communicate his block to other miners. He continues mining on his private chain and strategically discloses a subset of his private chain conditional on the number of blocks found by the other miners who always append to the longest public chain. Precisely, his disclosure strategy is as follows. In case of failing to generate a two-block lead over the public chain, the player publishes the block found in period $t$ when the lead of his private chain drops to zero, i.e., when it has the same length as the public chain. In case of generating a private chain lead of two blocks or more, the selfish miner discloses the earlier blocks of the private chain so that from the moment of disclosure onward all other miners append to the last disclosed block of the chain mined by the selfish miner. This causes the miners who follow the Bitcoin mining protocol to waste their found blocks, increasing the share of blocks belonging to the selfish miner. While some blocks of the selfish miner are also lost, in some cases this is profitable.

The profitability of selfish mining crucially depends on the selfish miner's mining share $\alpha_i$ and the behavior of other miners in case of a tie in length between the public chain and the disclosed chain of the selfish miner. Following \cite{eyal2018majority}, let $\gamma \in[0,1]$ denote the probability of miners $j \neq i$ appending to the disclosed block of the selfish miner in case of a tie. Note that in our model, uniform tie-breaking via the public randomization device corresponds to $\gamma=\frac12$. \cite{eyal2018majority} establish the following result.
\begin{proposition*}[Eyal and Sirer]
    For a given $\gamma$, a selfish miner $i$ with computational power $\alpha_i$ achieves a payoff larger than $\alpha_i$ if $\alpha_i>\frac{1-\gamma}{3-2\gamma}$. 
\end{proposition*}
This result implies that in our setting $(\gamma = \frac12)$, the threshold above which selfish mining is a profitable deviation from the Bitcoin mining protocol is $\alpha_i>\frac14$. To see the intuition behind selfish mining, consider the extreme case of $\gamma=1$. In this case, even initially withholding a block found in period $t$ does not come with the risk of it not being included in the longest chain: if miner $i$ fails to generate a two-block lead, he publishes the block he found at time $t$ and all other miners subsequently append to his block, leaving stale the honest block found in period $t+1$. Thus, selfish mining results in a strictly higher payoff than honest mining for any mining share $\alpha_i>0$, given the extreme case of $\gamma=1$. 

The analysis by \cite{eyal2018majority} uncovered the following fundamental question for Bitcoin that remains open: can fair decentralized consensus be achieved in equilibrium? This question is of fundamental importance because the core value proposition of Bitcoin and blockchain in general is decentralization. 


\section{Centralized Equilibria}
We now turn to the equilibrium analysis of the  mining game. We first establish a negative result, showing that a centralized equilibrium always exists, under a genericity assumption on mining power.

\begin{proposition}
\label{prop:centralized}
Suppose that $\alpha_1 > \alpha_i$ for all $i>1$. Then there exists an equilibrium in which player $1$ asymptotically owns all the blocks in the longest chain. In particular $u_1=1$ and $u_i=0$ for all $i>1$.
\end{proposition}
The equilibrium underlying Proposition \ref{prop:centralized} is
straightforward. Each player mines exclusively on its own chain, publishes each newly mined block immediately, and ignores all blocks mined by other players. Because player~$1$ has strictly greater mining power than every other player, its chain eventually becomes, and thereafter remains, longer
than every competing chain with probability one. The longest chain therefore
asymptotically consists entirely of blocks mined by player~$1$. No player $i>1$ can profitably deviate, since other players never extend a chain following a block mined by player $i$.

Bitcoin was designed as a decentralized system in which no single entity controls the production of blocks or, consequently, the ordering of transactions. Proposition \ref{prop:centralized} shows, however, that centralized consensus is itself an equilibrium outcome whenever the mining-power profile has a unique largest miner, which is generically the case. Importantly, this conclusion does not require the largest miner
to control a majority of aggregate mining power. Of course, this is not a realistic equilibrium, since the other miners have no reason to stay in the game, and the value of Bitcoin would probably collapse. 

When miner $1$ has mining power greater than $1/2$, the situation is more severe. In this case every equilibrium is centralized, with miner $1$ owning (almost) all blocks in the longest chain. 
\begin{proposition}\label{proposition:onlymonopoly}
If $\alpha_1>\frac12$ then player $1$ asymptotically owns all the blocks in the longest chain in every equilibrium. In particular, $u_1=1$ and $u_i=0$ for all $i>1$ in every equilibrium.
\end{proposition}
The idea is again simple: in every strategy profile, player $1$ can deviate and mine his own chain, ignoring all other blocks. This guarantees him the maximal utility $u_1=1$, and thus every equilibrium must satisfy $u_1=1$. This result is straightforward and motivates our focus on the regime $\alpha_i<\frac12$ in the remainder of the paper.

Proposition \ref{proposition:onlymonopoly} also provides a game-theoretic counterpart to the Byzantine fault tolerance (BFT) perspective common in computer science. The BFT perspective takes the behavior of honest miners as given and asks whether consensus is robust to arbitrary behavior by the remaining mining power; for Nakamoto consensus, this yields the familiar one-half threshold. Instead, in our analysis all behavior is endogenous. Proposition \ref{proposition:onlymonopoly} shows that the same threshold has an equilibrium interpretation: once a single miner controls more than half of the mining power, decentralized consensus cannot be sustained in equilibrium.

\section{Inertial Mining: A Decentralized Consensus Equilibrium}
We propose a novel mining strategy that we call \emph{inertial mining}. Motivated by Proposition \ref{proposition:onlymonopoly}, we restrict attention to the case of $\alpha_i < 1/2$ for all $i$. The inertial mining strategy, which is parametrized by $I \in \mathbb{N}$, is defined as follows. 

\begin{itemize}
\item \textbf{Single chain.} If $P_{t-1}$ consists of a single chain, choose $b_t^i$ to be the label of the last block in this chain. 
\item Otherwise, let $\tilde C_t^1,\ldots,\tilde C_t^n$ be the chains in $P_{t-1}$, ordered by increasing lengths $\ell_1 \leq \ell_2 \leq \cdots \leq \ell_n$. 
\begin{itemize}
    \item \textbf{$I$-switch.} If $\ell_n \geq \ell_{n-1} + I$, choose $b_t^i$ to be the label of the last block in $\tilde C_t^n$. 
    \item \textbf{Inertia.} Otherwise, 
    \begin{itemize}
        \item \textbf{Direct continuation.} If there is a unique chain containing $b_{t-1}^i$, choose $b_t^i$ to be the label of the last block in this chain.
        \item \textbf{Randomized fork resolution.} If there are multiple chains containing $b_{t-1}^i$, choose one of them at random using the public coordination device $\Xi_t$, as follows: start at block $b_{t-1}^i$, choose one of its children uniformly at random, and continue doing this until reaching a block with no children. Set $b_t^i$ to that block.
    \end{itemize}
    
\end{itemize}
\item All blocks are published as soon as they are mined.
\end{itemize}
\bigskip

A few notes on inertial mining and its relation to the standard Bitcoin protocol are in order:
\begin{enumerate}
    \item On path, all blocks are published immediately, and so behavior is identical to the one resulting from the standard Bitcoin mining protocol.
    \item The difference between inertial mining and standard mining occurs once a deviation from the protocol is observed. 
    Off path, standard mining always switches to a longest public chain. Inertial mining instead remains on the branch it mined on in the previous period unless a competing chain is at least $I$ blocks longer.
    \item The case $I=1$ is similar to the standard Bitcoin protocol: if there is a unique longest chain then players mine that chain. However, tie-breaking works differently if there is no unique longest chain.
    \item It is straightforward to see that if $I$ is large enough then selfish mining is not profitable. The challenge is to prove that no other deviation is profitable; this is non-trivial, since the strategy space is rather rich. The analysis is further complicated by the asymptotic nature of the utility, which implies that no one-shot deviation principle holds.
    \item It is important that when publishing at time $t$, $\Xi_{t+1}$ is not yet known, so that tie-breaking is uniformly random, conditioned on the information available to the players.
\end{enumerate}

\begin{claim}
    \label{clm:on-path-payoffs}
    Under inertial mining, the utility of each player $i$ is almost surely $\alpha_i$.
\end{claim}
Since under inertial mining there is a unique chain that contains every mined block, Claim~\ref{clm:on-path-payoffs} follows immediately from the strong law of large numbers. In particular, this implies that the outcome is decentralized and fair, in the sense that all miners participate and receive a share of blocks equal to their mining power.

Our main result is the following:
\begin{theorem}\label{thm:equilibrium}
Given a mining power profile $\alpha$ with $\max_i \alpha_i<1/2$, the inertial mining protocol is an equilibrium of the mining game for sufficiently large $I$. 
\end{theorem}
Theorem \ref{thm:equilibrium} establishes that fair decentralized consensus can indeed be achieved as an equilibrium outcome in Bitcoin mining:  for any desired threshold of the maximal mining power, there exists a corresponding equilibrium parametrized by $I$ that achieves fair decentralized consensus. We leave for future work an explicit calculation of how large $I$ needs to be, but conjecture that this is small enough to be practical, assuming $\max_i\alpha_i$ is not extremely close to one half. We note that even for small $\max_i\alpha_i$ it is far from obvious that the standard Bitcoin protocol is an equilibrium of our mining game.\footnote{As mentioned above, this is known for the complete information model of \cite{kiayias2016blockchain}.} The remainder of this section is devoted to the key constructions and an outline of the proof.

Suppose all players are playing according to the inertial mining protocol, except perhaps player $i$ who is using some other strategy. We prove that player $i$'s utility is at most $\alpha_i$, which by Claim~\ref{clm:on-path-payoffs} shows that there is no profitable deviation. Without loss of generality we assume that player $i$'s strategy is not randomized, since if there is a mixed profitable deviation then there is also a pure one. We henceforth fix the strategies of all the players, including the deviating player $i$, and calculate probabilities and expectations with respect to the distribution over outcomes generated by these strategies, the randomness in the mining process, and the randomness in the tie-breaking correlation device.

Since all players other than $i$ behave identically, we will assume that there is only one player other than $i$, and denote this player $j$. This is done for notational convenience only. The mining power of this player is $1-\alpha_i$. Recall that $b_t^j$ is the block to which player $j$ is trying to mine a successor at time $t$. With more than two players, all players different from $i$ choose the same $b_t^j$, and so there is no loss of generality in assuming that there are only two players.

Denote by $H_t^i=(\Xi_1,\ldots,\Xi_{t},P_1,\ldots,P_{t-1},M_1^i,M_2^i,\ldots,M_{t-1}^i)$ the history observed by player $i$ at the beginning of period $t$. 

We say $m_t$, the block mined at time $t$, is \emph{dishonestly mined} if either it is not published at time $t$, or its predecessor is not $b_t^j$. Otherwise, we say $m_t$ is \emph{honestly mined}. Note that all blocks mined by $j$ are honestly mined. We say a dishonestly mined block is an \emph{initial dishonestly mined} (henceforth, initial) block if its predecessor is honestly mined.  

The \emph{depth} $\Delta(x,y)$ of a block $(x,y)$ that is a descendant of the genesis block $(0,0)$ is the length of the chain starting at $(0,0)$ and ending at $(x,y)$. It follows from the definition of inertial mining that $\Delta(b_1^j) \leq \Delta(b_2^j) \leq \cdots$, and that in periods in which $j$ mines a block this inequality is strict. Of course, it is also strict if $i$ mines honestly, and in some other cases. The following claim is an immediate consequence of the definitions, and is important in understanding the dynamics that result from  our assumption that player $j$ follows the protocol.
\begin{claim}
\label{clm:move-forward}
    Suppose that the block $m_t$ was honestly mined. Then $\Delta(b_{t+1}^j) > \Delta(b_t^j)$. It follows that blocks mined by $j$ have distinct depths. 
\end{claim}

Fix a time $T$. We will use lower-case $t$ to denote times up to and including time $T$. Let $C_T^1,\ldots,C_T^k \subseteq P_T$ be the longest chains in $P_T$. Denote by $C_T$ a uniformly random chain chosen from these $k$ chains.  Denote by $K^{j}_T = M^j_T \setminus C_T$ the set of \emph{killed blocks} at time $T$. These are the blocks of player $j$ that are not in the longest chain by time $T$. Since player $j$ follows the inertial mining protocol, he will always add blocks to the tree originating in the genesis block $(0,0)$. Hence each block $(x,y) \in K^{j}_T$ will have at least one ancestor in $C_T$. We denote by $a(x,y) \in C_T$ the closest ancestor of $(x,y)$ that is in $C_T$ (see Figure~\ref{fig:ancestor-cousin}). We note that a killed block at time $T$ may be included in $C_{T'}$ for some $T'>T$, so the notion of killing is time-dependent. 

Given a killed block $(x,y)$ that is a descendant of the genesis block, denote by $\cus(x,y)$ the block $(w,z)$ in $C_T$ satisfying $\Delta(x,y)=\Delta(w,z)$ (see Figure~\ref{fig:ancestor-cousin}).
\begin{claim}
  Let $(x,y) \in K^j_T$ be a killed block. The block $\cus(x,y)$  is dishonestly mined by $i$. 
\end{claim}
\begin{proof}
Note that if $m_t$ is honestly mined then $\Delta(m_t) = \Delta(b^j_t)+1$, regardless of whether it was mined by $i$ or $j$. It thus follows from Claim~\ref{clm:move-forward} that any two honestly mined blocks cannot have the same depth. Since every block in $K^j_T$ is honestly mined, it must be that $\cus(x,y)$ is dishonestly mined.
\end{proof}

Fix a parameter $I$, and also fix $J \in \{1,\ldots,I-1\}$. We will show that the conclusion of the theorem holds if $J$ is large enough, and $I-J$ is also large enough.

\begin{figure}[t]
\centering
\begin{tikzpicture}[x=1.2cm,y=0.95cm,every node/.style={font=\small}]

\tikzstyle{main}=[circle,draw=black,fill=white,inner sep=1.7pt]
\tikzstyle{killed}=[circle,draw=red!70!black,fill=red!20,inner sep=1.7pt]

\node[main,label=right:{$(0,0)$}] (g) at (0,0) {};
\node[main] (c1) at (0,1) {};
\node[main,label=right:{$a(x_3,y_3)$}] (a) at (0,2) {};
\node[main,label=right:{$(w,z)$}] (c3) at (0,3) {};
\node[main,label=right:{$\cus(x_1,y_1)$}] (c4) at (0,4) {};
\node[main,label=right:{$\cus(x_2,y_2)$}] (cus) at (0,5) {};
\node[main,label=right:{$\cus(x_3,y_3)$}] (c5) at (0,6) {};
\node[main] (c6) at (0,7) {};

\draw[thick] (g)--(c1)--(a)--(c3)--(c4)--(cus)--(c5)--(c6);

\node[right] at (-0.2,7.5) {$C_T$};

\node[killed,label=left:{$(x_0,y_0)$}] (b1) at (-1.4,3) {};
\node[killed,label=left:{$(x_1,y_1)$}] (b2) at (-1.4,4) {};
\node[killed,label=left:{$(x_2,y_2)$}] (x) at (-1.4,5) {};
\node[killed,label=left:{$(x_3,y_3)$}] (x1) at (-1.4,6) {};

\draw[thick] (a)--(b1)--(b2)--(x)--(x1);

\draw[dashed,gray] (b2)--(c4);
\draw[dashed,gray] (x)--(c3);
\draw[dashed,gray] (x1)--(c3);
\draw[dashed,gray] (b1)--(c3);

\end{tikzpicture}
\caption{Illustration of the definitions of $a(x,y)$, $\cus(x,y)$ and $\kill(x,y)$. The block $a(x_3,y_3)$ is the closest ancestor of $(x_3,y_3)$ that lies on the longest chain $C_T$. The block $\cus(x_i,y_i)$ is the unique block on $C_T$  with the same depth as $(x_i,y_i)$. Suppose $J=1$. Then $\kill(x_1,y_1)=\cus(x_1,y_1)$ and $(w,z)$ is the killer of $(x_0,y_0)$, $(x_2,y_2)$ and $(x_3,y_3)$.}
\label{fig:ancestor-cousin}
\end{figure}
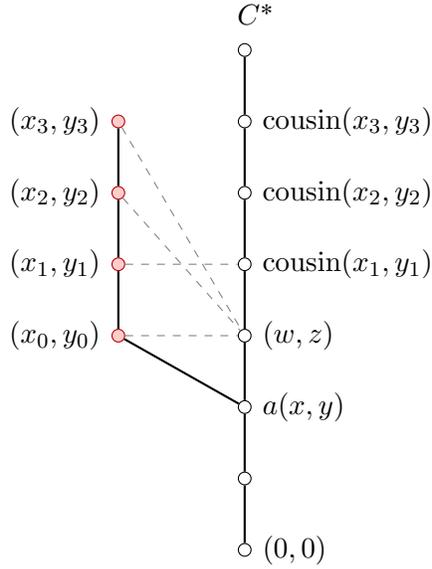

Suppose $(x,y) \in K^j_T$ is a killed block.  Let $(w,z)$ be the ancestor block of $\cus(x,y)$ such that all the blocks between $(w,z)$ and $\cus(x,y)$ (including $(w,z)$ itself) are dishonestly mined, yet the immediate predecessor of $(w,z)$ is honestly mined. Then $(w,z)$ must be an initial dishonestly mined block. We define the \emph{killer} of the killed block $(x,y)$ as follows. If $\Delta(\cus(x,y))\leq \Delta(w,z)+J$, then we set $\kill(x,y) = \cus(x,y)$ to be the killer of $(x,y)$. Otherwise, we set $\kill(x,y) = (w,z)$. See Figure~\ref{fig:ancestor-cousin}. 

The definition of $\kill(x,y)$ is the key to our proof: it sets up an accounting system that attributes every killed block to some dishonestly mined killer. The proof then shows that no block can be expected to kill enough blocks to make dishonest mining worthwhile. 

Given $t \leq T$, let $k_t$ be the number of blocks killed up to time $T$ by $m_t$, the block mined at time $t$:
\begin{align*}
    k_t = |\{(x,y) \in K^j_T \,:\, \kill(x,y) = m_t\}|.
\end{align*}

Let $\ell_t = 1_{m_t \in M^i_t \cap C_T}$ be the indicator of the event that the block mined at time $t$ will belong to $i$ and is included in the selected longest chain $C_T$. 

Define the random variable 
\begin{equation*}
S_t =\alpha_i k_t +(1-\alpha_i )\ell_t,
\end{equation*}
which we call the \emph{impact} of the block generated at time $t$. This turns out to be a useful metric of this block's usefulness to player $i$, capturing both the block's contribution to $i$'s chain, and the damage it does to $j$'s chain.

The next proposition is the key technical component in the proof of Theorem~\ref{thm:equilibrium}.
\begin{proposition}\label{prop:QR}
If $J$ and $I-J$ are sufficiently large, then there exists a nonnegative
summable sequence $(r_n)_{n\geq0}$ such that for every $T$ and
every $t\leq T$,
\begin{equation*}
\E{S_t}{H_{t}^{i}}\leq \alpha_i
(1-\alpha_i )+r_{T-t}
\end{equation*}
almost surely.
\end{proposition}
That is, given any history, the expected impact is at most $\alpha_i(1-\alpha_i)$ plus some small term.  

This proposition illustrates a tradeoff between killing others' blocks and retaining one's own blocks in the longest chain. A strategy that attempts to exclude others' blocks from the longest chain must carry a nontrivial risk of player $i$'s own 
blocks being excluded instead. In the proof of Theorem~\ref{thm:equilibrium} we use this to show that, in expectation, player $i$ cannot increase the fraction of his blocks in the longest chain no matter what he does. From this, we derive the conclusion of the theorem, namely that there are no profitable deviations.

\section{Discussion and Conclusion}
In this paper, we revisit a paradox:  Bitcoin's core distinguishing feature is that it is decentralized, yet the Bitcoin protocol---the mining strategy proposed by Satoshi Nakamoto and adopted by Bitcoin miners---is subject to profitable deviations \citep{eyal2018majority}. Importantly, the existence and nature of decentralized equilibria remained an open question.  This is not purely a theoretical concern: \cite{li2024statistical} provide empirical evidence consistent with selfish mining in several proof-of-work blockchains, most prominently in Monacoin and Bitcoin Cash. The lack of a known equilibrium for proof-of-work blockchains that generates a single longest chain on the equilibrium path that satisfies fair decentralized consensus has been highlighted by Ethereum's founder  \citep{buterin2017posfaq} and by \cite{hall2024study} in the context of Ethereum's transition from proof-of-work to proof-of-stake. 

Our main contribution is to resolve this paradox by introducing inertial mining and establishing that, for a sufficiently large $I$, it constitutes an equilibrium whenever every miner controls less than one half of the mining power. Moreover, inertial mining satisfies fair decentralized consensus, i.e., every miner's share of blocks coincides with his mining power, and it results in a single chain.  This resolves a central vulnerability of Bitcoin’s consensus mechanism: the existence of selfish mining, a profitable deviation from the standard Bitcoin mining protocol.


Our model of Bitcoin mining is highly stylized and simplified. A number of features can be added to make it more realistic. First, instead of having  exactly one block mined in each period, one could fix a small $\eps>0$, and have each player $i$ independently mine a block with probability $\eps\alpha_i$. This would correspond to each period representing a very short amount of time, with a small probability of mining in each period. We believe that our results would apply equally under this modification, at the cost of making the model less tractable and the proofs even more technical.

Another natural modification would be to introduce network outages. In a network outage, the group of miners is split into two, and communication between the groups is stopped for some random amount of time, so that blocks published by members of a group are only observed in that group. Network outages do occur and result in naturally occurring forks that are not a result of deviations from the protocol. The Bitcoin protocol was designed with the intention of being robust to them, and we believe that inertial mining can likewise be robust to network outages.  Nevertheless, adding network outages is a more significant change to the model, and furthermore a change that can be implemented in many ways. We again believe that the spirit of our results should survive this type of alteration.

A practical concern about inertial mining is that most network outages are short and result in forks of length one, and resolving them by waiting for a lead of $I$ is inefficient. One could modify inertial mining to allow for random tie-breaking resolution in this case of a length-one fork, instead of waiting for a lead of length $I$. We believe that our proof applies equally to this variant.


As a direction for future research, we note that proof-of-stake chains suffer from incentive problems that are similar to the selfish-mining deviations of Bitcoin \citep{brown2019formal}. It would be interesting to understand if a similar approach can be applied there to design equilibrium protocols. 





\bibliography{refs}

\appendix

\section{Appendix: Missing Proofs}
\subsection{Proof of Proposition~\ref{prop:centralized}}
\begin{proof}
For each player $i$, consider the following \emph{independent strategy}: if player $i$ previously mined a block, $M_t^i\ne \emptyset$ he appends to the last of his mined blocks and immediately publishes a block. If player $i$ has not previously mined a block, $M_t^i = \emptyset$, he appends to the genesis block $(0,0)$. Fix an player $i>1$ and assume all other players follow the independent strategy. For any deviation of player $i$, all other players mine their own chain and ignore the blocks mined by player $i$. The result relies on the following fact which we will prove below: For any strategy of player $i$, the longest chain will be asymptotically solely mined by player $1$.

For each player $k$, let
\[
    N_k(t):=\left|M_t^k\right|
\]
be the number of blocks he finds in the first $t$ periods. By the strong law of large numbers, 
\[
    N_k(t)=\alpha_k t+o(t)
\]
almost surely. Define
\[R_t:=\max_{k}\max_{0\leq s\leq t}|N_k(s)-\alpha_ks|=o(t).\]
By the independent strategy, any chain contains blocks of at most one player $k\neq i$, followed possibly by a sequence of blocks mined by $i$. Suppose that for $k\neq 1$, a chain $C$ at time $t$ consists of a prefix of player $k$'s chain, whose final block was mined at time $s$, followed by blocks mined by $i$. Then
\begin{equation*}
    |C|
    \leq
    N_k(s)+N_i(t)-N_i(s).
\end{equation*}
A chain containing only blocks of $i$ satisfies the same bound with the first term equal to zero. Then
\begin{align*}
   \max_{k\neq 1} \max_{0\leq s\leq t}
    \bigl\{N_k(s)+N_i(t)-N_i(s)\bigr\}
    &\leq  3R_t+ \max_{k\neq 1}\alpha_ks+\alpha_it-\alpha_is\\
    &\leq 3R_t+\max_{k\neq 1}\alpha_kt
\end{align*}
On the other hand, player $1$'s publicly available chain, denoted $\bar C_t^1$, has length
\[
    |\bar C_t^1|=N_1(t)\geq\alpha_1t-R_t.
\]
As player $1$ has strictly higher mining power than other players, $\bar C_t^1$ will eventually be longer than any such chain not containing any block of player $1$.

Let $t$ be sufficiently large such that every longest chain at time $t$ must contain blocks of player $1$. Let $C_t^l$ be a longest public chain at time $t$ and suppose its last block belonging to player $1$ is mined at time $s$. Such a chain can be longest only if
\[
    N_1(s)+N_i(t)-N_i(s)\geq  |\bar C_t^1|=N_1(t).
\]
This requires 
\[\alpha_1 s+\alpha_i t-\alpha_i s +3R_t\geq \alpha_1 t -R_t,\]
which we write as 
\[(\alpha_1-\alpha_i)(t-s)\leq 4R_t=o(t).\]
Hence, $t-s=o(t)$.
Therefore, 
\[|C_t^l\cap M_t^i|\leq N_i(t)-N_i(s)\leq \alpha_i(t-s)+2R_t=o(t).\]
That is uniformly over all the longest chains, the fraction of $i$'s blocks goes to $0$. Therefore, $u_i=0$, and there is no profitable deviation from the proposed strategy. 

Under the proposed profile player 1 obtains utility one, which is the maximal possible utility, so player 1 cannot profitably deviate.

\end{proof}

\subsection{Proof of Proposition \ref{proposition:onlymonopoly}}
\begin{proof}
Fix arbitrary strategies for all players other than player $1$. We claim that the following strategy for player $1$ yields utility $1$: he always mines at his most recently mined block and publishes every block immediately. This will imply that in every equilibrium, $u_1=1$ and the fraction of blocks in any longest chain mined by player $1$ converges to $1$. 

Let $N_1(t)$ be the number of blocks player $1$ gets in the first $t$ periods, and let $N_{-1}(t)$ be the number of blocks mined by all other players. By the strong law of large numbers, 
\[R_t:=\max_{0\leq s\leq t}|N_1(s)-\alpha_1 s|=\max_{0\leq s\leq t}|N_{-1}(s)-(1-\alpha_1) s|=o(t)\qquad a.s.,\]
where the first equality is because $N_1(s)+N_{-1}(s)=s$.

By the mining strategy of player $1$, any chain consists of a prefix of player $1$'s chain (possibly empty) followed by blocks mined by the other players. Let $C_t^l$ be a longest public chain at time $t$ and let $s$ be the period in which the last block of player $1$ contained in $C_t^l$ was mined. Then
\[N_1(t)\leq N_1(s)+N_{-1}(t)-N_{-1}(s).\]
This implies 
\[\alpha_1 t-R_t\leq \alpha_1 s+(t-s)(1-\alpha_1)+3R_t.\]
Rearranging yields, 
\[(2\alpha_1-1)(t-s)\leq 4R_t=o(t).\]
It follows that $t-s=o(t)$.

Therefore, the number of blocks mined by other players in $C_t^l$ is at most $t-s=o(t)$. On the other hand, the length of $C_t^l $ is at least $N_1(t)=\alpha_1t+o(t)$, the length of the chain entirely built by player $1$. Thus, uniformly over all longest chains, the fraction of blocks not mined by player~$1$ converges to zero. This implies that $u_1=1$ and all other players have utility $0$.
\end{proof}

\subsection{Proof of Theorem~\ref{thm:equilibrium}}
We first show that the longest chain grows at a linear rate. 
\begin{lemma}
\label{lm:linear-growth-Ct}
\[
    \liminf_{T\to\infty}\frac{|C_T|}{T}\geq 1-\alpha_i \qquad \text{almost surely.}
\]
\end{lemma}
\begin{proof}
    By Claim~\ref{clm:move-forward}, when player $j$ mines a block at time $t$, $\Delta(b_{t+1}^j)>\Delta(b_t^j)$. Hence,
    \[|C_T|\geq \Delta(b_{T+1}^j)\geq |M_T^j| .\]
By the strong law of large numbers,
\[
    \frac{|M_T^j| }{T}\to 1-\alpha_i
\]
almost surely. Hence
\[
    \liminf_{T\to\infty}\frac{|C_T|}{T}
    \geq
    \lim_{T\to\infty}\frac{|M_T^j| }{T}
    =
    1-\alpha_i
\]
almost surely.
\end{proof}

We need to show that $\E{u_i} \leq \alpha_i$. We claim that it suffices to show that $\Pr{u_i > \alpha_i} = 0$. 
\begin{claim}
    \label{clm:as}
    Suppose that $i$ has a strategy such that $\E{u_i} > \alpha_i$. Then $i$ has a (perhaps different) strategy such that $\Pr{u_i > \alpha_i}=1$.
\end{claim}
The idea of the proof is that since the payoff does not depend on what fraction of blocks the player had at any finite time, the player can always ``start from scratch'' and draw a new $u_i$ if they think it is unlikely to be above $\alpha_i$.
\begin{proof}[Proof of Claim~\ref{clm:as}]
Suppose that $\Pr{u_i > \alpha_i} = p > 0$, and let $A$ be the event $\{u_i > \alpha_i\}$. Then $p_t=\E{A}{H_t^i}$ is a bounded martingale, and converges almost surely to the indicator of the event $A$. Consider the following alternative strategy for player $i$: follow the original strategy, unless there is some time $s$ such that $p_s < p/2$. If no such time exists then $\lim_t p_t=1$, the event $A$ occurs and the utility is strictly above $\alpha_i$. If there is a time $s$ such that $p_s < p/2$, the player ``starts from scratch'', implementing the original strategy but treating the block $b_s^j$ as the genesis block. By Lemma \ref{lm:linear-growth-Ct}, $|C_t|\rightarrow\infty$ so that for any time $t'$ the payoff can also be written as
\begin{align*}
    u_i = \liminf_t  \frac{1}{\#\text{ of longest chains at time }t }\sum_{C_t: \text{ a longest chain at time $t$}} \frac{|C_t \cap (M_t^i \setminus M_{t'}^i)|}{|C_t|}.
\end{align*}
In this expression, we only include blocks mined after time $t'$, and yet the payoff is the same, because the blocks before time $t'$ form a vanishing fraction of the total blocks.

Repeating this whenever $\Pr{u_i > \alpha_i}{H_t^i}$ goes below $p/2$ yields, by the law of large numbers, a strategy such that $u_i > \alpha_i$ almost surely. 
\end{proof}

\medskip

\begin{proof}[Proof of Proposition~\ref{prop:QR}]
We write $\alpha:=\alpha_i$ throughout the proof. We fix $t\leq T$.

Let $Z_n$ be the difference between the numbers of blocks mined by $i$ and $j$ in the $n$ periods following period $t$. Then $Z_n$ can be written as 
\[
    Z_n:=\sum_{\ell=1}^n X_\ell,
    \qquad
    \Pr{X_\ell=1}=\alpha,
    \qquad
    \Pr{X_\ell=-1}=1-\alpha,
\]
where the $X_\ell$'s are independent, and capture the identity of the player in periods $t+1,\ldots,t+n$ who mined a block. Note that process $(Z_n)_n$ is a sum of i.i.d.\ random variables, or a random walk. It is a biased random walk, with drift to the left: $\E{Z_n} = n(2\alpha-1) < 0$, and so it is unlikely to ever go far to the right. A standard calculation shows that the probability that $Z_n$ is ever large is 
\begin{align}
  \label{eq:rw}
  \Pr{\text{there exists n such that } Z_n \geq k} = \left(\frac{\alpha}{1-\alpha}\right)^k.
\end{align}

If $m_t$ is honestly mined, then $m_t$ cannot kill any block, in which case $S_t=(1-\alpha)\ell_t\leq 1-\alpha$. Henceforth $m_t$ is dishonestly mined. 

Since $S_t=0$ unless player $i$ mines in period $t$, it is enough to prove
\[
    \E{S_t\mid H_t^i,m_t\in M^i_t}\leq 1-\alpha+r_{T-t}
\]
for some summable $(r_n)_{n\geq0}$. 

Let $w_t$ be the closest initial dishonest ancestor of $m_t$ (which possibly coincides with $m_t$). 
By definition, $m_t$ can possibly kill a block if 
\begin{equation}\label{eq:J-close}
     \Delta(m_t)\leq\Delta(w_t)+J.
\end{equation}
If $m_t$ does not kill a block, then $S_t=(1-\alpha)\ell_t\leq 1-\alpha$, and so it suffices to consider the case that $m_t$ does kill at least one block. Accordingly, suppose henceforth that \eqref{eq:J-close} holds. Since the predecessor of $w_t$ is honestly mined, and since the $\Delta(b_s^j)$ is non-decreasing in $s$,
\[
    \Delta(b_t^j)\geq\Delta(w_t)-1.
\]
Consequently,
\begin{equation}
    \Delta(m_t)-\Delta(b_t^j)\leq J+1.
    \label{eq:initial-depth-bound}
\end{equation}

We let
\[
\begin{split}
    k_t^{\mathrm{cous}}
    &:=
    \mathbf 1_{\{
       \text{$m_t$ is the cousin and killer of some block}
    \}},\\
    k_t^{\mathrm{far}}
    &:=
    k_t-k_t^{\mathrm{cous}},\\
    Q_t&:=(1-\alpha) \ell_t+\alpha k_t^{\mathrm{cous}},\\
    R_t&:=\alpha k_t^{\mathrm{far}}.
\end{split}
\]

Since, by assumption, $m_t$ is within $J$ of $w_t$, it can kill one block by being its cousin (accounted for by $k_t^{\mathrm{cous}}$). If $m_t=w_t$ is initial, it can kill additional blocks that are farther down the chain (accounted for by $k_t^{\mathrm{far}}$).

Using this notation, 
\[S_t=Q_t+R_t.\]
\textbf{Far killing}.
$k_t^{\mathrm{far}}=0$ unless $m_t$ is initial. Supposing that $m_t$ is initial,  $m_t$ can kill at most one block at each depth $d\geq \Delta(m_t)+J+1$. 

Let $k\geq J+1$. If a block of depth $\Delta(m_t)+k$ is killed, let $c\in C_T$ be the cousin of the killed block. By the definition of killer, all blocks from $m_t$ to $c$ are dishonestly mined and hence are mined by $i$. 

If player $j$ ever begins mining on a descendant of $c$ by time $T$, consider the first such time $t'$. Before this time, the branch containing $c$ is only extended by player $i$, while every block mined by player $j$ strictly increases the depth of blocks $j$ mines at. Since player $j$ never switches to a shorter chain, player $i$ must have mined at least as many blocks as player $j$ in the time interval $[t,t']$. Thus, $1+Z_{t'-t}\geq 0$. If player $j$ never begins mining on a descendant of $c$ by time $T$,
the same inequality holds at $t'=T$, because the branch
containing $c$ is a longest chain at time $T$. 

Summing over the blocks of depth at least $\Delta(m_t)+J+1$, we have 
\begin{equation}
\begin{split}
    \E{R_t\mid H_t^i,m_t\in M^i_t}
    &\leq
    \alpha\sum_{k=J+1}^{\infty}\mathbb{P}(1+Z_{t'-t}\geq 0\text{ for some }t'\geq k)
\end{split}
    \label{eq:far-charge}.
\end{equation}
Let $r_J'$ be this upper bound, which is independent of $T$, since $(Z_n)_n$ is simply a sum of i.i.d.\ random variables. As $\mathbb{P}(1+Z_{t'-t}\geq 0\text{ for some }t'\geq k)$ decays exponentially fast in $k$, $r_J'\rightarrow 0$ as $J\rightarrow\infty$.

\paragraph{Survival and same-depth killing.}

Let
\[
    \tau:=\inf\{s>t:
        b_s^j=m_t
        \text{ or }b_s^j\text{ is a descendant of }m_t\}
\]
be the first period after $t$ in which player $j$ mines at a block that is descendant of $m_t$ (or $m_t$ itself). If this set is empty, let $\tau=\infty$.  There are three ways in which $j$ starts to mine at some descendant of $m_t$ before time $T$, i.e. $\tau\leq T$. 

\begin{enumerate}
\item \textbf{Randomized fork resolution.}
Let $\mathcal{T}$ be the event that $m_t$ goes through a randomized fork resolution (i.e., player $j$ uses the public correlation device to randomize the choice between the chain containing $m_t$ and other chains of the same length), and let $\mathcal{T}'$ be the sub-event that $m_t$ wins the one-depth fork resolution. Conditional on $\mathcal{T}$, the chain containing $m_t$ is selected with probability at most $1/2$, by the definition of the tie-breaking distribution on chains. On $\mathcal{T}'$, we have the trivial bound $Q_t\leq 1$.

\item \textbf{$I$-switch.}
Let $\mathcal{I}$ be the event that player $j$'s mining target first becomes a descendant of $m_t$ because they find a chain containing $m_t$ that is at least $I$ longer than the next longest chain. Suppose such a switch occurs at time period $t+n$. Up to that time, every block on the adopted chain after $m_t$ is mined by $i$. Every block mined by $j$ strictly increases where player $j$ mines at. Hence, 
\[Z_n\geq I-(\Delta(m_t)-\Delta(b_t^j))\geq I-J-1,\]where the last inequality follows from the condition
\eqref{eq:initial-depth-bound}. By the standard biased random walk estimate \eqref{eq:rw}, 
\begin{align}
    \mathbb{P}(\mathcal I\mid H_t^i,m_t\in M^i_t)
    &\leq\mathbb{P}(\text{there exists }n\text{ such that }Z_n\geq I-J-1)\nonumber\\
    &=\left(\frac{\alpha}{1-\alpha}\right)^{I-J-1}.
    \label{eq: bound I}
\end{align}
Importantly, this is small when $I-J$ is large. Again, on the event of $\mathcal I$, we have the trivial bound $Q_t\leq 1$.

\item \textbf{Direct continuation.}
Let $\mathcal{U}$ be the event that player $j$ first mines at $m_t$ or its descendant by direct continuation—that is, without a tie-breaking step and without an $I$-switch. Since $m_t$ is dishonestly mined, if player $j$ gets the block at time $t+1$, it necessarily creates a fork, so adoption of the branch containing $m_t$ can no longer be uncontested. Therefore, 
\begin{equation}
    \mathbb{P}(\mathcal U\mid H_t^i,m_t\in M^i_t)
    \leq\alpha.
    \label{eq:uncontested-bound}
    \end{equation}
On $\mathcal U$, there is no block of $j$ at depth $\Delta(m_t)$ so that $k_t^{\mathrm{cous}}=0$ and 
    \[    Q_t\leq 1-\alpha.\]

\end{enumerate}

Next, we consider the case $\tau>T$, i.e., the case that $j$ does not adopt $m_t$ by time $T$. Let
\[
    \mathcal E   := \{\tau>T,\ m_t\in C_T\}.
\]
On this event, every block after $m_t$ is mined by player $i$ by time $T$. Since the branch $C_T$ containing $m_t$ is longest at time $T$, \eqref{eq:initial-depth-bound} implies 
\[Z_{T-t}\geq-(J+1). \]
Let
\[r_n:=\mathbb{P}(Z_{n}\geq -(J+1)).\]
Then
\begin{equation}
\mathbb{P}(\mathcal E\mid H_t^i,m_t\in M^i_t)\leq r_{T-t}.
\label{eq:bound E}
\end{equation}
Moreover, $r_n$ decays exponentially fast in $n$. In particular, $(r_n)_n$ is summable. On the event of $\mathcal E$, we have the trivial bound $Q_t\leq 1$.

Summing up the events above, we reach the bound
\[
    Q_t
    \leq
    (1-\alpha)\mathbf 1_{\mathcal U}
    +\mathbf 1_{\mathcal T'}
    +\mathbf 1_{\mathcal I}
    +\mathbf 1_{\mathcal E}.
\]
By definition, $\mathcal{T}$ and $\mathcal{U}$ are disjoint so that
\begin{equation*}
    \begin{split}
        &\E{(1-\alpha)\mathbf 1_{\mathcal U}
    +\mathbf 1_{\mathcal T'}\mid H_t^i,m_t\in M^i_t}\\
\leq\,&\E{(1-\alpha)\mathbf 1_{\mathcal U}
    +\frac{1}{2}\mathbf 1_{\mathcal T}\mid H_t^i,m_t\in M^i_t}\\
     \leq\,&\E{(1-\alpha)\mathbf 1_{\mathcal U}
    +\frac{1}{2}\mathbf 1_{\mathcal U^c}\mid H_t^i,m_t\in M^i_t}\\
    \leq\,& (1-\alpha)\alpha+\frac{1}{2}(1-\alpha)\\
    =\,&(1-\alpha)(\alpha+\frac{1}{2})
    \end{split},
\end{equation*}
where we use \eqref{eq:uncontested-bound} for the last inequality. 
Together with \eqref{eq: bound I} and \eqref{eq:bound E},
 \begin{equation}
 \begin{split}
    \E{Q_t\mid H_t^i,m_t\in M^i_t}
    &\leq
    (1-\alpha)(\alpha+\frac{1}{2})+\E{\mathbf 1_{\mathcal I}\mid H_t^i,m_t\in M^i_t}+\E{\mathbf 1_{\mathcal E}\mid H_t^i,m_t\in M^i_t}\\
     &\leq (1-\alpha)(\alpha+\frac{1}{2})+(\alpha/1-\alpha)^{I-J-1}+r_{T-t}.
     \label{eq:local-charge}
     \end{split}
\end{equation}

Combining \eqref{eq:far-charge} and \eqref{eq:local-charge},
\[\E{S_t\mid H_t^i,m_t\in M^i_t}\leq r_J'+ (1-\alpha)(\alpha+\frac{1}{2})+(\alpha/1-\alpha)^{I-J-1}+r_{T-t}.\]
Choose $J$ and $I-J$ large enough such that 
\[r_J'+(\alpha/1-\alpha)^{I-J-1}\leq (1-\alpha)(\frac{1}{2}-\alpha),\]
then
\[\E{S_t\mid H_t^i,m_t\in M^i_t}\leq 1-\alpha+r_{T-t}.\]
We are able to conclude from here as $S_t=0$ if player $i$ does not mine a block in period $t$. 
\end{proof}

We will need the following general lemma regarding the limit of the ratio of two stochastic processes.
\begin{lemma}\label{lm: liminfratio}
Suppose that $\{X_t\}$ and $\{Y_t\}$ are two random sequences and $a>0$ is a constant such that 
\begin{enumerate}
    \item $X_t,Y_t\in [0,1]$.
    \item $\liminf_t\E{aY_t-X_t}\geq 0$.
    \item $\liminf_t\E{Y_t}> 0$.
\end{enumerate}
Then it is impossible that
\[\liminf_t \frac{X_t}{Y_t}>a \;\text{ a.s.}\]
\end{lemma}

\begin{proof}[Proof of Lemma~\ref{lm: liminfratio}]
Define $Z_t := X_t - a Y_t$. Because $X_t \geq 0$ and $Y_t \leq 1$, we have a deterministic lower bound: $Z_t \geq -a Y_t \geq -a$. 

Assume for the sake of contradiction that
\begin{align*}
    \Pr{\liminf_{t \to \infty} \frac{X_t}{Y_t} > a} = 1.
\end{align*} 
Let $L = \liminf_{t \to \infty} \frac{X_t}{Y_t}$ and define the random variable $\varepsilon := \frac{1}{2} \min(1, L - a)$. Then $\varepsilon > 0$ almost surely, and there is a random variable $N$ such that for $t \geq 0$ we have $\frac{X_{N+t}}{Y_{N+t}} \geq a + \varepsilon$. We can equivalently write this as
\begin{equation*}
    Z_{N+t}=X_{N+t} - a Y_{N+t} \geq \varepsilon Y_{N+t}.
\end{equation*}
We then have for all $t$ that
\begin{equation*}
    Z_t \geq \varepsilon Y_t \mathbf{1}_{\{t \geq N\}} - a \mathbf{1}_{\{t < N\}}.
\end{equation*}
Taking the expectation and then $\limsup$ of both sides and using $(2)$ in the assumptions,
we obtain:
\begin{equation*}
    0 \geq \limsup\E{Z_t} \geq \limsup\left(\E{\varepsilon Y_t \mathbf{1}_{\{t \geq N\}}} - a \mathbb{P}(N > t)\right)=\limsup \E{\varepsilon Y_t \mathbf{1}_{\{t \geq N\}}}.
\end{equation*}
As $Y_t\geq 0$, this implies that $\lim_{t \to \infty} \E{\varepsilon Y_t \mathbf{1}_{\{t \geq N\}}} = 0$. In particular, $\varepsilon Y_t \mathbf{1}_{\{t \geq N\}}$ converges to $0$ in probability. It follows that $Y_t$ converges to $0$ in probability. Because the sequence $Y_t$ is uniformly bounded, it guarantees convergence in expectation:
\begin{equation*}
    \lim_{t \to \infty} \E{Y_t} = 0.
\end{equation*}
However, this contradicts (3) in our initial hypothesis. Therefore, it is impossible that $\liminf_{t \to \infty} \frac{X_t}{Y_t} > a$ almost surely.
\end{proof}

\begin{proof}[Proof of Theorem~\ref{thm:equilibrium}]

We first choose $J$ and $I-J$ large enough for the conclusions from Proposition \ref{prop:QR} to be true. 

We let $C_T^1,\ldots,C_T^{n_T}$ be all the longest chains at time $T$. Recall that $C_T$ is the uniformly random chain chosen from these $n$ chains. Set
\[X_T=\frac{1}{T\cdot n_T}\sum_{l=1}^{n_T}|C_T^l\cap M_T^i|,\qquad Y_T:=\frac{1}{T}|C_T|,\qquad \alpha:=\alpha_i\]
and check the conditions of Lemma \ref{lm: liminfratio}. 
The first condition in the lemma is trivially satisfied as $T=|M_T|\geq |M_T^i|$. 

We next verify the second condition. By definition,
\[
    |C_T\cap M_T^i|=\sum_{s=1}^T \ell_{s}.
\]
Moreover, every block in $M_T^j\setminus C_T$ is assigned to a
unique killer block mined at some time $s\leq T$. Hence
\[
    |C_T\cap M_T^j|
    =
    |M_T^j|-\sum_{s=1}^T k_{s}.
\]
It follows that
\begin{align*}
    \E{\alpha Y_T-X_T}
    &=    
    \E{\alpha  \frac{1}{T\cdot n_T}\sum_{l=1}^{n_T}|C_T^l\cap M_T^j|-\frac{1}{T\cdot n}\sum_{l=1}^{n_T}(1-\alpha)|C_T^l\cap M_T^i|} \\
    &=
    \frac{1}{T}
    \E{\alpha |C_T\cap M_T^j|-(1-\alpha)|C_T\cap M_T^i|} \\
    &=
    \frac{1}{T}
    \E{\alpha\left(|M_T^j|-\sum_{s=1}^T k_{s}\right)
    -(1-\alpha)\sum_{s=1}^T \ell_{s}} \\
    &=
    \frac{1}{T}
    \sum_{s=1}^T
    \left[
    \alpha(1-\alpha)
    -
    \E{\alpha k_{s}+(1-\alpha)\ell_{s}}
    \right] \\
    &=
    \frac{1}{T}
    \sum_{s=1}^T
    \left[
    \alpha(1-\alpha)-\E{S_{s}}
    \right].
\end{align*}
By Proposition~\ref{prop:QR} and the law of iterated expectations,
\[
    \E{S_{s}}\leq \alpha(1-\alpha)+r_{T-s}.
\]
Therefore
\[
    \E{\alpha Y_T-X_T}
    \geq
    -\frac{1}{T}\sum_{s=1}^T r_{T-s}
    =
    -\frac{1}{T}\sum_{n=0}^{T-1}r_n.
\]
Since $(r_n)_n$ is summable,
\[
    \liminf_{T\to\infty}\E{\alpha Y_T-X_T}\geq 0.
\]
Lastly, by Lemma \ref{lm:linear-growth-Ct}, $|C_T|$ has a linear growth rate so that $\liminf_t\E{Y_t}> 0$. 

Lemma~\ref{lm: liminfratio}  states that it is impossible that 
\[\liminf_t \frac{X_t}{Y_t}>\alpha \quad a.s.,\]
which translates to the impossibility of
\[
    u_i=\liminf_{T\to\infty}\frac{1}{n_T}\sum_{l=1}^{n_T}
    \frac{|C_T^l\cap M_T^i|}{|C_T^l|}
    >
    \alpha_i \quad a.s.
\]

By Claim \ref{clm:as}, it is impossible that $\E{u_i}> \alpha_i$ as desired. 
\end{proof}
\end{document}